\documentclass{article}
\usepackage[margin=1.2in]{geometry}
\usepackage{amsmath,amsthm,xspace}
\newtheorem{theorem}{Theorem}

\newcommand{\aest}{{\bf A}}
\newcommand{\best}{{\bf B}}
\newcommand{\cest}{{\bf C}}
\newcommand{\nexptime}{\textsc{NExpTime}\xspace}

\makeatletter
\newcommand{\manuallabel}[2]{\def\@currentlabel{#2}\label{#1}}
\makeatother

\begin{document}

\title{A note on the product homomorphism problem and CQ-definability%
\thanks{We are grateful to Ross Willard for discussions on the topic
  and 
  for comments on an earlier draft.}}

\author{Balder ten Cate and Victor Dalmau}

\maketitle
 
The \emph{product homomorphism problem} (PHP) takes as input a finite
collection of relational structures $\aest_1, \ldots, \aest_n$ and another
relational structure $\best$, all over the same schema, and asks whether
there is a homomorphism from the direct product $\aest_1\times\cdots
\times \aest_n$ to $\best$. This problem is clearly solvable in
non-deterministic exponential time. It follows from results in
\cite{Willard10} that the problem is \nexptime-complete. The proof,
based on a reduction from an exponential tiling problem, uses
structures of bounded domain size but with 
relations of unbounded arity.  In this note, we provide a
self-contained proof of \nexptime-hardness of PHP,
and we show that it holds already for directed
graphs, as well as for structures of
bounded arity with a bounded domain size (but without a bound on the
number of relations). More precisely, we obtain:

\begin{theorem}
  \manuallabel{thm:binary}{\thetheorem.1}
  \manuallabel{thm:Willard}{\thetheorem.2}
  \manuallabel{thm:digraphs}{\thetheorem.3}
The PHP is \nexptime-complete~\cite{Willard10}. The lower bound
holds already for
\begin{enumerate}
\item structures with binary relations and a bounded domain size;
\item structures with a single relation and a bounded domain size;
\item structures with a single binary relation
\end{enumerate}
\end{theorem}
This completes the picture, since PHP is solvable in polynomial time when all three of the
above parameters (i.e., number of relations, arity, and domain size) are bounded, as follows from the fact that, in this case, there are only finitely
many different possible input structures up to isomorphism. 

 Theorem~\ref{thm:binary} is proved by an adaptations of the technique
used in \cite{Willard10}. Theorem~\ref{thm:Willard} is proved by a
reduction
from~\ref{thm:binary}. Theorem~\ref{thm:digraphs} is proved
by a reduction 
%to
from 
Theorem~\ref{thm:Willard}.

We also present an application of the above result to the
CQ-definability problem (also known as the PP-definability problem).

\section{Proof of Theorem~\ref{thm:binary}}

\begin{proof}
    By reduction from the exponential tiling problem. We are assuming a fixed set of tile types with associated horizontal and vertical compatibility relations, and the input of the 
    tiling problem consists of an integer $m$ (specified in unary) together with a 
  sequence of (not necessarily distinct) tile types $t_1, \ldots, t_m$. The problem is to
  decide whether the $2^m$-by-$2^m$ grid has a valid tiling where $t_1, \ldots, t_m$
   is a prefix of the sequence of tiles on the first row, starting at
   the origin.
  It is known that there is a fixed finite set of tile types for which
  this problem is \nexptime-hard.

  For ease of exposition, we will also make use of unary
  relations. These can easily be replaced by binary ones.

   The idea of the reduction is very simple.     We will define $2m$ structures, $\aest_1, \ldots, \aest_{2m}$,
   each having domain $\{0,1\}$. In this way, each element of the product $\Pi_i \aest_i$ is a bitstring of length $2m$, which we will interpret as a pair of bistrings of length $m$, where the first bitstring is the binary encoding of a horizontal coordinate and the second bitstring is the binary encoding of a vertical coordinate. The structure $\best$ will have one element for each tile type, so that a map $h:\Pi_i \aest_i\to \best$ can be viewed as a way of assigning a tile type to each position on the $2^m$-by-$2^m$ grid. Furthermore, by endowing the structures involved with suitable relations, we will ensure that every homomorphism $h:\Pi_i \aest_i\to \best$ corresponds to a valid tiling, and vice versa.

    Let $H, V\subseteq (\{0,1\}^{2m})^2$ be the horizontal and vertical successor relations on coordinate
    pairs.
    In other words, $H = \{(\textbf{x}\textbf{y}, \textbf{x}'\textbf{y})\mid
    \textbf{x}'=\textbf{x}+1\}$ and $V = \{(\textbf{x}\textbf{y}, \textbf{x}\textbf{y}')\mid
    \textbf{y}'=\textbf{y}+1\}$.
 Let $P_0, \ldots, P_m$ be singleton sets denoting the coordinate pairs
    $(0,0), \ldots, (m-1, 0)$. In order to make our reduction work, we need to somehow make
sure that the relations $H, V, P_0, \ldots, P_m$ are ``available'' in the product structure $\Pi_i \aest_i$, by choosing the factor structures $\aest_1, \ldots, \aest_{2m}$ appropriately.

Let us say that an $n$-ary relation $R$ over the domain $\{0,1\}^{2m}$ is \emph{decomposable} if it can be represented as a product $R_1\times\cdots\times R_{2m}$ where each $R_i$ is an $n$-ary relation over $\{0,1\}$. Intuitively, this means that
if we include in each factor structure $\aest_i$ the relation $R_i$, then
the product structure $\Pi_i \aest_i$ contains the relation $R$. 
Each of the unary relations $P_0, \ldots, P_m$, being a singleton, is trivially decomposable.
Indeed, if we use the notation $\underline{k}[i]$ to denote the value of the $i$-th bit of the binary encoding of the number $k$, then $P_k = \{\underline{k}[1]\}\times\cdots\times \{\underline{k}[m]\}\times \{0\}^m$.
The binary relations $H$ is \emph{not} decomposable. However, it turns out to be a union of decomposable relations, which will suffice for our purposes. First, observe that whenever $(\textbf{x}\textbf{y},\textbf{x}'\textbf{y}')\in H$ then 
   $\textbf{x}' = \textbf{x}+1$, and therefore $\textbf{x}$ must have at least one bit that is set to zero. For each $k\leq m$, let $H_k$ be the subrelation of $H$ containing all $(\textbf{x}\textbf{y},\textbf{x}'\textbf{y}')\in H$ for which it is the case that 
the $k$-th bit of $\textbf{x}$ is the least significant bit that is 0.
By definition, we have that $H = \bigcup_k H_k$. Then 
    $H_k$ decomposes: $H_k = \textsf{id}^{k-1}\times\textsf{diff}^{m-k+1}\times\textsf{id}^m$,
    where \textsf{id} is the identity relation on $\{0,1\}$ and \textsf{diff} is the difference relation on $\{0,1\}$.
The exact same story holds for $V$
(where we have that $V_k = \textsf{id}^m \times \textsf{id}^{k-1}\times\textsf{diff}^{m-k+1}$).

We are now ready to defined the structures $\aest_1, \ldots, \aest_{2m}$ and $\best$.
The signature consists of the relations $H_1, \ldots, H_m, V_1, \ldots, V_m, P_0, \ldots P_m$. 
    For $m,\ell<k$, we define 
    \[P_k^{\aest_\ell} = \begin{cases} \{\underline{k}[\ell]\} & \text{if $\ell\leq m$} \\
                                                 \{0\} & \text{otherwise} 
                           \end{cases}
   \]
    \[H_k^{\aest_\ell} = \begin{cases} \textsf{diff} & \text{if $\ell\in [k,m]$} \\
                                                  \textsf{id} & \text{otherwise} 
                           \end{cases}
    \]
    \[V_k^{\aest_\ell} = \begin{cases}\textsf{diff} & \text{if $\ell\in [m+k, 2m]$} \\
                                                 \textsf{id} & \text{otherwise} 
                           \end{cases}
    \]

   The structure $\best$ is defined as follows:
   its domain is the set of all tile types. The unary predicate $P_i$ 
  denotes the singleton set $\{t_i\}$ as specified in the instance of the tiling problem.
   The relations $H_k$ and $V_k$ contain all pairs of tile types that are horizontally,
    respectively vertically, compatible.

    It is now straightforward to verify that there is a homomorphism $h:\Pi_i \aest_i \to \best$
   if and only if there is a valid tiling.
\end{proof}

\section{Proof of Theorem~\ref{thm:Willard}}

\begin{proof}
 The proof proceed by a reduction from Theorem~\ref{thm:binary},
 which states that PHP is in \nexptime even
  for structures of a bounded domain size.  The reduction goes in two steps.
  We first reduce to the case with two
  relations. Let $\best$ be any structure  with domain $D$ and with multiple relations
  $R_1,\dots,R_k$ of respective arity $r_1, \ldots, r_k$ over $D$.
  We denote by $\best^*$ the structure with domain $D\cup\{0\}$
  that has 
\begin{enumerate}
\item[(i)] a unary relation $P$ denoting the set $D$
\item[(ii)] a relation $R$ of arity
   $r_1+ \cdots + r_k$ consisting of 
   the all-zeroes tuple $(0, \cdots, 0)$, and, 
   for every $(a_1\ldots a_{r_i})\in R_i$ ($1\leq i\leq
  k$), the tuple whose first $r_1 + \cdots + r_{i-1}$ coordinates
   are all $0$, whose subsequent $r_i$ coordinates are $a_1\ldots
   a_{r_i}$,  and whose final $r_{i+1}+\cdots+r_k$ coordinates are $0$
   again.
\end{enumerate}
This transformation can be carried out in
  polynomial time, and it increases the domain of each structure with
  at most one element. Furthermore, we claim that
  $\Pi_i\aest_i^*\to \best^*$ if and
  only if $\Pi_i\aest_i\to \best$. 
  In one direction, suppose $h:\Pi_i\aest_i^*\to\best^*$. By
  construction (and, more specifically, due to the presence of the
  unary relation $P$), $h$ must map every element of $\Pi_i\aest_i$ to
  an element of $\best$. It is then easy to see that $h$ is in fact a
  homomorphism from $\Pi_i\aest_i$ to $\best$. 
  Conversely, suppose $h:\Pi_i\aest_i\to
  \best$. Let $h'$ be the map from $\Pi_i\aest_i^*$
  to $\best^*$ that extends $h$ such that every element of $\Pi_i\aest_i^*$
  containing a 0 is sent to the element 0 of $\best^*$. 
  Then $h'$ is a homomorphism from 
  $\Pi_i\aest_i^*$ to $\best^*$. This follows from 
  the fact that (i) no element containing a 0 can 
  belong to the $P$ relation in $\Pi_i\aest_i^*$, 
  and (ii) if a tuple in the relation $R$ of $\Pi_i\aest_i^*$
  includes an element containing a $0$, then this
  tuple consists entirely of elements that contain a 0, 
  and hence $h'$ maps the tuple in question to 
  the all-zeroes tuple, which belongs again to $R$ in $\best^*$.

  As a final step, we further reduce to the case with a single relation.
  This is done by replacing each structure with two relations, $P$ and
  $R$, by the structure with the same domain and with a single
  relation that is defined as the cartesian product of $P$ and $R$. 
  Again, this transformation can be carried out in polynomial time,
  it does not affect the domains of the structures involved, and
  it preserves the existence or non-existence of a homomorphism
  from $\Pi_i\aest^*_i$ to $\best^*$.  
\end{proof}

\section{Proof of Theorem~\ref{thm:digraphs}}

\begin{proof}
  We shall give a reduction from the PHP with a single
  relation (Theorem~\ref{thm:Willard}). Let $\aest_1,\dots,\aest_n,\best$ be
  an instance of the PHP, using a single $r$-ary relation $R$. We may
  assume without loss of generality that, for each structure
  $\cest=(C,R^\cest)$  among  $\aest_1,\dots,\aest_n,\best$, the projection of
  $R^\cest$ to the first coordinate is the entire domain $C$.
  This is because we can always replace the $r$-ary relation $R$ by
  the
  $r+1$-ary relation $C\times R$ where $\times$
  indicates here the cartesian product.
 This
  transformation can be carried out in polynomial time and it does not affect
  the existence or non-existence of a homomorphism from
  $\Pi_i\aest_i$ to $\best$.
Henceforth, we shall use $R$ to denote the unique relation, and $r$ to
denote its arity. 
If $t$ is a $r$-ary tuple and $j\in\{1,\dots,r\}$ we shall denote by $t[j]$ the $j$th component of $t$.

For every $i$, we define $G(\aest_i)$ to be the following digraph:

The nodes of $G(\aest_i)$ include all elements of $\aest_i$. Furthermore,
for every tuple $t=(a_1,\dots,a_r)\in R_i$, $G(\aest_i)$ contains $r$
additional nodes, which we
denote by $t^j$ with $j=1,\dots,r$.
 These nodes are connected by
the following directed edges:
\begin{itemize}

\item $(t^j,t^{j+1})$ for every $1\leq j<r$.

\item $(t[j],t^j)$ for every $1\leq j\leq r$.

\end{itemize}

We define $G(\best)$ as  the digraph obtained from $\best$ in the
same way, except that we further add 
$r-1$ additional elements $s^1,\dots,s^{r-1}$ called {\em sink nodes},
connected by edges $(s^j,s^{j+1})$ for every $1\leq j< r-1$, and 
an edge from every element of $\best$ to every sink node.

\begin{trivlist}
\item \textbf{Claim:} there is a homomorphism $h:\Pi_i G(\aest_i)\to G(\best)$ if and only if
  there is a homomorphism $h':\Pi_i\aest_i\to\best$.
\end{trivlist}

In the remainder, we prove this claim,
which immediately implies the theorem. 
We start with the more difficult direction: let $h$ be a homomorphism from $\Pi_i \aest_i$ to $\best$.  We shall define from $h$ a homomorphism $h'$ from $\Pi_i G(\aest_i)$ to $G(\best)$. Let $v=(v_1,\dots,v_n)$ be a node of  $\Pi_i G(\aest_i)$.

\begin{itemize}

\item If $v_i\in \aest_i$ for all $i$ then we say that $v$ is of ``type $1$''. In this case we define $h'(v)=h(v)$.

\item If, for all $i$, $v_i=t_i^{j_i}$ where $t_i$ is a tuple in (the relation of) $\aest_i$
and $j_i\in\{1,\dots,r\}$ then:

\begin{itemize}

\item If, in addition, there exists some $j$ such that $j_i=j$ for every $i$ then we say that $v$ is of ``type $2$''.  Note that $t_1\times\cdots\times t_m$ is a tuple in 
$\Pi_i \aest_i$ and hence $h(t_1\times\cdots\times t_m)$ (where $h$ is applied component-wise) is a tuple of $\best$. In this case, define $h'(v)$ to be $h(t_1\times\cdots\times t_m)^j$.

\item Otherwise we say that $v$ is of ``type $3$'' and we set $h(v')$
  to the sink node $s^j$ where $j=\min{j_i}$. Observe that, in this case, necessarily $j\leq r-1$.

\end{itemize}

\item If $v$ is not in any of the previous types then we say that is of ``type $4$''. In this case, we shall prove there exists a vertex $u$ of type $1$ such that for 
every vertex $w$ of type $2$ the following holds: 
$$(v,w) \text{ is an edge of } \Pi_i G(\aest_i) \Rightarrow (u,w) \text{ is an edge of } \Pi_i G(\aest_i)$$
In this case we set $h'(v)=h'(u)$.
Let us show that such $u$ exists. If there exists $i,i'$ such that $v_i=t_i^{j_i}$ and $v_{i'}=t_{i'}^{j_{i'}}$ and $j_i\neq j_{i'}$ then clearly $v$ does not have an outgoing edge to any vertex of type $2$ and we can set $u$ to be any arbitrary vertex of type $1$. Same applies if there exist $i$ such that $v_i=t_i^r$. Consequently we are left with the case in which there exists some $j\in\{1,\dots,r-1\}$ such that for every $i$, $v_i\in \aest_i$ or $v_i=t_i^{j}$ for some tuple $t_i$ in $\aest_i$. Define $u_i$ to be $v_i$ in the first case and $t_i[j+1]$ in the second and set $u=(u_1,\dots,u_m)$.

Let $w=(w_1,\dots,w_n)$ be a node of type $2$. We shall prove that for every $i$, if $(v_i,w_i)$ is an edge of $\Pi_i G(\aest_i)$ then so
if $(u_i,w_i)$. The claim is obvious whenever $u_i=v_i$. 
Assume now that $v_i=t_i^j$. Since $t_i^j$ has only one outgoing edge (to $t_i^{j+1}$) in $G(\aest_i)$ it follows that 
$w_i=t_i^{j+1}$. The claim follows from the fact that $u_i=t_i[j+1]$ and $G(\aest_i)$ contains edge $(t_i[j+1],t_i^{j+1})$. 
\end{itemize}

Let us prove that $h'$ is indeed a homomorphism. Let $(u,v)$ be an edge in $\Pi_i G(\aest_i)$ and let $u=(u_1,\dots,u_m)$ and $v=(v_1,\dots,v_m)$. We shall prove 
that $(h(u),h(v))$ belongs to $G(\best)$ by means of a case analysis on the types of $u$ and $v$. Notice that $v$ is necessarly of type $2$ or $3$ since nodes of type $1$ or $4$ do not have incoming edges.

\begin{itemize}

\item $u$ is of type $1$. If $v$ is of type $3$ the claim follows from the fact that $G(\best)$ has an edge from every element in $\best$ to every sink vertex. Assume now that $v$ is of type $2$, that is, of the form $(t_1^j,\dots,t_m^j)$.
Since $(u,v)$ is an edge of $\Pi_i G(\aest_i)$ and $u$ is of type $1$ it follows that
$u_i=t_i[j]$ for every $i$. Hence $u=(t_1\times\cdots\times t_m)[j]$ and, since $h$ defines a homomorphism, 
$h(u)$ is the $j$th component of $h(t_1\times\cdots\times t_m)$ ($h$ is applied component-wise). It follows that $G(\best)$ contains the edge from $h'(u)$ to
$h'(v)=h(t_1\times\cdots\times t_m)^j$.

\item $u$ is of type $2$. Then necessarily there exists $t_1,\dots,t_m$ and $j$ such that $u=(t_1^j,\dots,t_m^j)$ and $v=(t_1^{j+1},\dots,t_m^{j+1})$ and the claim follows directly from the definitions.

\item $u$ is of type $3$ then $v$ is necessarily of type $3$ as well. Furthermore, it follows that if $h'(u)$ is $s^j$ then necessarily $h'(v)=s^{j+1}$.

\item $u$ is of type $4$. It follows directly from the definition of $h'(u)$ and the fact that every vertex of type $3$ is mapped by $h'$ to a sink node.
\end{itemize}

Conversely, let $h'$ be a homomorphism from $\Pi_i G(\aest_i)$ to
$G(\best)$. Recall that each element of $\Pi_i \aest_i$ is in
particular an element of $\Pi_i G(\aest_i)$.  We claim that the
restriction of $h'$ to $\Pi_i \aest_i$ is a
homomorphism from $\Pi_i \aest_i$ to $\best$. 

First, we show 
that, for each element $t$ of $\Pi_i \aest_i$, $h'(t)$ is an
element of $\best$.
Let $t=(t_1, \ldots, t_n)$ be any element of $\Pi_i \aest_i$.  Recall
that we have assumed that the projection of $R^{\aest_i}$ on the first
coordinate is the entire domain of $\aest_i$. Hence, each $t_i$ is the
first component of some tuple in $R^{\aest_i}$. By construction of
$G(\aest_i)$, this implies that $t_i$ has an outgoing path of length
$r$ in $G(\aest_i)$, and hence, $t$ has an outgoing path of length $r$
in $\Pi_i \aest_i$.  It follows that $h'$ must map $t$ to a node of
$G(\best)$ that has an outgoing path of length $r$. By construction
of $G(\best)$, then, $h(t)$ must be an element of $\best$.

Next, we shall show that $h:\Pi_i\aest_i\to \best$ is a
homomorphism. Let
$(t^1, \ldots, t^r)\in R^{\Pi_i\aest_i}$, where each $t^j = (t^j[1],
\ldots, t^j[n])$. Then we have that
$(t^1[i],\ldots,t^r[i])$ belongs to $R^{\aest_i}$,  for each $i\leq n$.
Consequently, $(t^1[i],\ldots,t^r[i])$ satisfies the conjunctive
query $$q(x_1, \ldots, x_r)=\exists y_1\ldots y_r (\bigwedge_{1\leq
  i\leq r}
E(x_i,y_i)\land\bigwedge_{1\leq i<r} E(y_i,y_{i+1}))$$
It follows that $(t^1, \ldots, t^r)$ satisfies the same conjunctive
query in $\Pi_i G(\aest_i)$, and therefore, since conjunctive
queries are preserved by homomorphisms, 
$h(t^1, \ldots, t^r)$ satisfies $q$ in $G(\best)$. It follows by
construction of $G(\best)$ that $h(t^1, \ldots, t^r)\in R^\best$.
\end{proof}

\section{Application: CQ-definability}

The \emph{CQ-definability problem} (also known under the name
PP-definability, and several other names), is the problem with input an instance $I$ and a relation $S$
over the domain of $I$, to decide whether there is a conjuctive query
$q$ such that $q(I)=S$.  It has been long known that this problem is
decidable in co\nexptime (see discussion and references in
\cite{Willard10}). It was shown in \cite{Willard10} that the
CQ-definability problem is co\nexptime-complete, even for instances of
a bounded domain size. On the other hand, the proof used relations of
arbitrarily large arity. We show that  the same problem is
co\nexptime-complete for a fixed schema (but without a bound on the
size of the domains of the instances).

\begin{theorem}
 The CQ-definability problem is
  co\nexptime-hard already for unary queries over a fixed schema
  consisting of a single binary relation. 
\end{theorem}

\begin{proof}
 Reduction from PHP with a single binary relation $R$
  (Theorem~\ref{thm:digraphs}).  Let instances $\aest_1,
 \ldots, \aest_n$
  and $\best$ be given. Inspection of the proof of
  Theorem~\ref{thm:digraphs} shows that we may assume that, in each
  of these stuctures, the maximum length of a directed path is precisely $r$,
  for some fixed natural number $r$.  Let $\cest$ be the instance consisting
  of the disjoint union of $\aest_1, \ldots, \aest_n$ and $\best$, extended with
  the facts $R(a_i,x)$ for all $i\leq n$ and $x\in \aest_i$, and $R(b,x)$
  for all $x\in \best$, where $a_1, \ldots, a_n$ and $b$ are fresh
  elements. Observe that each $a_i$, and also $b$, by construction, has an outgoing
  path of length $r+1$, while no other elements have an outgoing
  path of length $r+1$.
  Let $S=\{a_1, \ldots,
  a_n\}$.  Then we claim that $\aest_1 \times ... \times \aest_n \to \best$ if and
  only if $S$ is not definable inside $\cest$ by a conjunctive query. In one direction,
  if $\aest_1 \times\ldots\times \aest_n \to \best$ then clearly $S$ is not
  definable by a conjunctive query, because, by homomorphism preservation, the same
  conjunctive query would have to select $b$.  On the other hand, if 
  $\aest_1 \times ... \times \aest_n \not\to \best$, then we can
  construct a query $q$ defining $S$ as follows: first we define $q_1$
  to be the canonical Boolean  conjunctice query of $\aest_1 \times\ldots\times
  \aest_n$, and we define $q(x)$ to be the unary conjunctive query
  expressing that $q_1$ holds in the submodel of $\cest$ consisting of
  all elements reachable (in one step) from $x$. By
  construction,  $q(\aest)$ includes all of $S$ and excludes $b$. It
  is also easy to see that $q(\aest)$ contains no elements other than
  $a_1, \ldots, a_n$ and $b$. Therefore, $q$ defines $S$. 
\end{proof}

\bibliographystyle{plain}
\bibliography{prod}
\end{document}